\documentclass{amsart}
 
\usepackage{amsfonts}
\usepackage{amsmath}
\usepackage{amssymb}
\usepackage{graphicx}

\usepackage{enumerate}

\newtheorem{theorem}{Theorem}
\theoremstyle{plain}

\newtheorem{claim}{Claim}

\newtheorem{definition}{Definition}

\newtheorem{lemma}{Lemma}

\newtheorem{remark}{Remark}

\numberwithin{equation}{section}

\begin{document}
\title[]{Symbols of non-archimedean Elliptic Pseudo-differential Operators, Feller Semigroups, Markov Transition Function and Negative Definite Functions}
\author{Ismael Guti\'{e}rrez Garc\'{\i}a}
\address{Universidad del Norte \\Department of Mathematic and Statistic}
\email{isgutier@uninorte.edu.co}
\author{Anselmo Torresblanca-Badillo}
\address{Universidad del Norte \\Department of Mathematic and Statistic}
\email{atorresblanca@uninorte.edu.co}
\subjclass{}
\keywords{Pseudo-differential operators, Feller semigroups, Markov
transition function, convolution semigroup, Negative definite function,
non-archimedean analysis. }

\begin{abstract}
In this article we prove that the heat kernel attached to the  non-archimedean elliptic pseudodifferential operators determine a Feller semigroup and a uniformly stochastically continuous $C_{0}-$transition function of some strong Markov processes $\mathfrak{X}$ with state space $\mathbb{Q}_{p}^{n}.$ We explicitly write the Feller semigroup and the Markov transition function associated with the heat kernel. Also, we show that the symbols of these pseudo-differential operators are a negative definite function and moreover, that this symbols can be represented as a combination of a constant $c\geq 0,$ a continuous homomorphism $l: \mathbb{Q}_{p}^{n}\rightarrow \mathbb{R}$ and a non-negative, continuous quadratic form $q: \mathbb{Q}_{p}^{n}\rightarrow \mathbb{R}.$
\end{abstract}

\maketitle

\section{Introduction}

The applications of $p$-adic analysis in mathematical physics has received much attention in the last decades due to the interest of studying $p-$adic pseudodifferential equations associated with certain physical models, see e.g. \cite{Albeverio 2006}, \cite{Chuong-Co-2008}, \cite{Khrennikov-1992}-\cite{R-Zu}, \cite{To-Z}-\cite{Zu-lib1}.

Pseudo-differential operators that have a closed extension that generates a Feller semigroup constitute a classical area of research in the archimedean setting see, e.g., \cite{Hoh-Libro}-\cite{Jacob-vol-3}. This fact aroused great interest in the non-archimedean sense of finding pseudo-differential operators that had a closed extension that would generate Feller's semigroups. The first and most recent work obtained about in the non-archimedean sense as a result of this interest is the paper \cite{To-Z-2}, in which a huge class of pseudodifferential operators (with negative definite symbols) is introduced.

It should be noted that all the above mentioned Feller semigroups are
implicitly expressed. 

In this article we consider the elliptic pseudo-differential operators $f(%
\mathcal{\partial },\beta )$ introduced by Z\'{u}\~{n}iga-Galindo, see \cite{Zu-2003},\cite{Zu-2004}, which has the form
\begin{equation}
(f(\mathcal{\partial },\beta )\varphi
)(x):=\mathcal{F}_{\xi\rightarrow x }^{-1}(|f(\xi )|_{p}^{\beta
}\mathcal{F}_{x\rightarrow \xi }\varphi ),\text{ }\varphi \in \mathcal{D}(\mathcal{\mathbf{\mathbb{Q}}}_{p}^{n})\text{, }\beta >0,  \label{operator}
\end{equation}%
where $f(\xi )\in \mathcal{\mathbf{\mathbb{Z}}}_{p}^{n}[\xi _{1},\ldots ,\xi _{n}]$ is an elliptic polynomial of degree $d
$ satisfying $f(\xi )=0$ if and only if $\xi =0.$

It is important to keep in mind that the symbol $|f(\xi )|_{p}^{\beta },$ $%
\beta >0,$ of these elliptic pseudo-differential operators are not
necessarily radial functions on $\mathcal{\mathbf{\mathbb{Q}}}_{p}^{n}$.

We study the following Cauchy problem
\begin{equation*}
\left\{
\begin{array}{ll}
\frac{\partial u}{\partial t}(x,t)=-(f(\mathcal{\partial },\beta )u)(x,t),%
\text{ \ } & t\in \lbrack 0,\infty ),\text{ \ }x\in \mathbb{Q}_{p}^{n} \\
&  \\
u(x,0)=u_{0}(x)\in \mathcal{D}(\mathcal{\mathbf{\mathbb{Q}}}_{p}^{n})\text{,} &
\end{array}%
\right.
\end{equation*}
naturally associated with these operators.

From the mathematical point of view, in the study of heat conduction and
diffusion, the heat kernel
\begin{equation*}
Z(x,t):=Z_{t}(x):=\mathcal{F}_{\xi\rightarrow x}^{-1}(e^{-t|f(\xi
)|_{p}^{\beta }})=\int\nolimits_{\mathbf{\mathbb{Q}}_{p}^{n}}\chi _{p}(-x,\xi )e^{-t|f(\xi )|_{p}^{\beta }}d^{n}\xi ,\text{ }%
t>0,\text{ }x\in \mathbf{\mathbb{Q}}_{p}^{n},
\end{equation*}%
attached to operator $f(\mathcal{\partial },\beta ),$ is the fundamental
solution to the heat equation with appropriate initial conditions. These
equations deal with problems that have spatial and temporal structure, that
is, their solution depends of a position $x$ and a time $t.$

Physically, the heat kernel represents the evolution of temperature in a
region whose boundary is held at a particular temperature (typically zero),
such that an initial unit of heat energy is placed at a point at time $t=0,$
in other words, the heat kernel $Z(x,t)$ expresses a thermal distribution of
position $x$ at time $t$. Inspired by this fact, unlike the previously
mentioned works, in this article we will obtain explicitly a Feller
semigroups $\{T_{t}\}_{t\geq 0}$ on the space of Banach $C_{0}(\mathbb{Q}_{p}^{n})$ (space of continuous functions vanishing at infinity) generated
from the heat kernel, see Theorem \ref{Feller_semigroups}. Moreover, we also
obtain in an explicit way a uniformly stochastically continuous $C_{0}-$%
transition function on $\mathbf{\mathbb{Q}}_{p}^{n},$ $p_{t}(x,\cdot ),$ that satisfies the condition: for each $s>0$
and each compact subset $E\subset\mathbb{Q}_{p}^{n}$,%
\begin{equation*}
\lim_{x\rightarrow \infty }\sup_{0\leq t\leq s}p_{t}(x,E)=0.
\end{equation*}%
Moreover, $p_{t}(x,\cdot )$ it is the transition function of some strong
Markov processes $\mathfrak{X}$ with state space $\mathbb{Q}_{p}^{n}$ and transition function $p_{t}(x,\cdot )$ whose paths are right
continuous and have no discontinuities other than jumps, see Theorem \ref{Transition}. \ \ \

It is relevant to mention that the type of Feller semigroups and transition functions treated in this article are not obtained from a closed extension of the elliptic pseudo-differential operator $f(\mathcal{\partial },\beta )$, as it was developed in the works mentioned above, on the contrary, our Feller semigroups and transition functions are obtained from the heat kernel
attached to operator $f(\mathcal{\partial },\beta ),$ which is very
important, because we are interested in the Markovian behavior of our process, therefore the interest to know explicitly the Feller semigroup and the transition function of our strong Markov process.

There are several significant differences between this article and the works mentioned above.

On the other hand, motivated by the article \cite{To-Z-2}, we are interested
in knowing if the symbol of the pseudo-differential operator $f(\mathcal{\partial },\beta )$ is a negative definite function and if there is also a
representation for said symbol.

This article is organized as follows. In Section \ref{Fourier Analysis}, we will collect some basic results on the $p$-adic analysis and fix the
notation that we will use through the article. In Section \ref%
{preliminary_results}, we give several technical results of elliptic
polynomials of degree $d$ and on the heat kernel $(Z_{t}(x))$ of the Cauchy problem associated with the elliptic pseudo-differential operator $f(\mathcal{\partial },\beta )$ of degree $d.$ In Section \ref{Feller and Transition}, we initially prove that there is a Feller semigroup $\{T_{t}\}_{t\geq 0}$ associated with the heat kernel $(Z_{t}(x))$, see Theorem \ref{Feller_semigroups}$.$ This semigroup of Feller is obtained explicitly and moreover is conservative, i.e. $T_{t}1=1,$ for all $t>0,$ see Remark \ref{Conservative}. Later in this section, we prove that there is an uniformly stochastically continuous $C_{0}-$transition function $p_{t}(x,\cdot )$ on $\mathbb{Q}_{p}^{n},$ written explicitly for $E\in \mathcal{B}(\mathbb{Q}_{p}^{n})$ as
\begin{equation*}
p_{t}(x,E)=\left\{
\begin{array}{ll}
Z_{t}(x)\ast 1_{E}(x)\text{,} & \text{\ for }t>0\text{, }x\in \mathbf{\mathbb{Q}}_{p}^{n} \\
&  \\
1_{E}(x), & \text{for }t=0\text{, }x\in \mathbf{\mathbb{Q}}_{p}^{n}.%
\end{array}%
\right.
\end{equation*}

Moreover, this transition function satisfies that for each $s>0$ and each
compact subset $E\subset\mathbb{Q}_{p}^{n}$,%
\begin{equation*}
\lim_{x\rightarrow \infty }\sup_{0\leq t\leq s}p_{t}(x,E)=0.
\end{equation*}

The Feller semigroup $\{T_{t}\}_{t\geq 0}$ and the transition function $p_{t}(x,\cdot )$ are connected by the equation
\begin{equation*}
T_{t}f(x):=\int_{\mathbf{\mathbb{Q}}_{p}^{n}}p_{t}(x,d^{n}y)f(y).
\end{equation*}
We also have that $p_{t}(x,\cdot )$ is the transition function of some
strong Markov processes $\mathfrak{X}$ with state space $\mathbb{Q}_{p}^{n}$ whose paths are right continuous and have no discontinuities other
than jumps, see Theorem \ref{Transition}. In Section \ref{Negative_definite_function}, a first important result obtained is that the symbol $|f|_{p}^{\beta },$ $\beta >0,$ of the pseudo-differential operator $f(\mathcal{\partial },\beta )$ is a negative definite function, see Theorem \ref{negative_function}. This result complements the family of non-archimedean pseudo-diferential operators treated at \cite{To-Z-2} whose symbol is a negative definite function. Subsequently, making use of the theory of convolution semigroup of local type, we can show that the symbol $|f|_{p}^{\beta },$ $\beta >0,$ can be represented in the form
\begin{equation*}
|f(\xi )|_{p}^{\beta }=c+il(\xi )+q(\xi )\text{, }\xi \in\mathbb{Q}_{p}^{n},
\end{equation*}%
where $c\geq 0,$ $l: \mathbb{Q}_{p}^{n}\rightarrow \mathbb{R}$ is a continuous homomorphism and $q: \mathbb{Q}_{p}^{n}\rightarrow \mathbb{R}$ is a non-negative, continuous quadratic form. This is because we will prove the equivalence of the following statements:

\begin{enumerate}[(i)]
\item for all open neighbourhoods $W$ of $0$ we have
\begin{equation*}
\lim_{t\rightarrow 0^{+}}\frac{1}{t}Z_{t}(\complement W)=0.
\end{equation*}%

\item $Z=0$ $(Z$ is the L\'{e}vy measure for the convolution semigroup $%
(Z_{t})_{t>0}$ on $\mathbb{Q}_{p}^{n}).$

\item $|f(\xi )|_{p}^{\beta }=c+il(\xi )+q(\xi )$ for $\xi \in\mathbb{Q}_{p}^{n}$, see Theorem \ref{equivalencias}.
\end{enumerate}

\section{\label{Fourier Analysis}Fourier Analysis on $\mathbb{Q}_{p}^{n}$: Essential Ideas}

\subsection{The field of $p$-adic numbers}

Along this article $p$ will denote a prime number. The field of $p-$adic numbers $\mathbb{Q}_{p}$ is defined as the completion of the field of rational numbers $\mathbb{Q}$ with respect to the $p-$adic norm $|\cdot |_{p}$, which is defined as
\begin{equation*}
\left\vert x\right\vert _{p}=\left\{
\begin{array}{lll}
0 & \text{if} & x=0 \\
&  &  \\
p^{-\gamma } & \text{if} & x=p^{\gamma }\frac{a}{b}\text{,}%
\end{array}%
\right.
\end{equation*}%
where $a$ and $b$ are integers coprime with $p$. The integer $\gamma :=ord(x) $, with $ord(0):=+\infty $, is called the\textit{\ }$p-$\textit{adic order of} $x$.

Any $p-$adic number $x\neq 0$ has a unique expansion of the form
\begin{equation*}
x=p^{ord(x)}\sum_{j=0}^{\infty }x_{j}p^{j},
\end{equation*}%
where $x_{j}\in \{0,1,2,\dots ,p-1\}$ and $x_{0}\neq 0$. By using this
expansion, we define \textit{the fractional part of }$x\in \mathbb{Q}_{p}$,
denoted $\{x\}_{p}$, as the rational number
\begin{equation*}
\left\{ x\right\} _{p}=\left\{
\begin{array}{lll}
0 & \text{if} & x=0\text{ or }ord(x)\geq 0 \\
&  &  \\
p^{ord(x)}\sum_{j=0}^{-ord_{p}(x)-1}x_{j}p^{j} & \text{if} & ord(x)<0.%
\end{array}%
\right.
\end{equation*}%
In addition, any non-zero $p-$adic number can be represented uniquely as $%
x=p^{ord(x)}ac\left( x\right) $ where $ac\left( x\right) =\sum_{j=0}^{\infty
}x_{j}p^{j}$, $x_{0}\neq 0$, is called the \textit{angular component} of $x$. Notice that $\left\vert ac\left( x\right) \right\vert _{p}=1$.

We extend the $p-$adic norm to $\mathbb{Q}_{p}^{n}$ by taking
\begin{equation*}
||x||_{p}:=\max_{1\leq i\leq n}|x_{i}|_{p},\text{ for }x=(x_{1},\dots
,x_{n})\in \mathbb{Q}_{p}^{n}.
\end{equation*}
We define $ord(x)=\min_{1\leq i\leq n}\{ord(x_{i})\}$, then $%
||x||_{p}=p^{-ord(x)}$.\ The metric space $\left(\mathbb{Q}_{p}^{n},||\cdot ||_{p}\right) $ is a complete ultrametric space, which is a totally disconnected topological space. For $r\in \mathbb{Z}$, denote by $B_{r}^{n}(a)=\{x\in 
\mathbb{Q}_{p}^{n};||x-a||_{p}\leq p^{r}\}$ \textit{the ball of radius }$p^{r}$
\textit{with center at} $a=(a_{1},\dots ,a_{n})\in\mathbb{Q}_{p}^{n}$, and take $B_{r}^{n}(0):=B_{r}^{n}$. Note that $B_{r}^{n}(a)=B_{r}(a_{1})\times \cdots \times B_{r}(a_{n})$, where $B_{r}(a_{i}):=\{x\in \mathbb{Q}_{p};|x_{i}-a_{i}|_{p}\leq p^{r}\}$ is the one-dimensional ball of radius $p^{r}$ with center at $a_{i}\in  \mathbb{Q}_{p}$. The ball $B_{0}^{n}$ equals the product of $n$ copies of $B_{0}= 
\mathbb{Z}_{p}$, \textit{the ring of }$p-$\textit{adic integers of }$\mathbb{Q}_{p}$. We also denote by $S_{r}^{n}(a)=\{x\in \mathbb{Q}_{p}^{n};||x-a||_{p}=p^{r}\}$ \textit{the sphere of radius }$p^{r}$ \textit{with center at} $a=(a_{1},\dots ,a_{n})\in \mathbb{Q}_{p}^{n}$, and take $S_{r}^{n}(0):=S_{r}^{n}$. We notice that $S_{0}^{1}= \mathbb{Z}_{p}^{\times }$ (the group of units of $\mathbb{Z}_{p}$), but $\left( \mathbb{Z}_{p}^{\times }\right) ^{n}\subsetneq S_{0}^{n}$. The balls and spheres are both open and closed subsets in $\mathbb{Q}_{p}^{n}$. In addition, two balls in $\mathbb{Q}_{p}^{n}$ are either disjoint or one is contained in the other.

As a topological space $\left(\mathbb{Q}_{p}^{n},||\cdot ||_{p}\right) $ is totally disconnected, i.e. the only
connected \ subsets of $\mathbb{Q}_{p}^{n}$ are the empty set and the points. A subset of $\mathbb{Q}_{p}^{n}$ is compact if and only if it is closed and bounded in $\mathbb{Q}_{p}^{n}$, see e.g. \cite[Section 1.3]{V-V-Z}, or \cite[Section 1.8]%
{Albeverio et al}. The balls and spheres are compact subsets. Thus $\left(\mathbb{Q}_{p}^{n},||\cdot ||_{p}\right) $ is a locally compact topological space.

We will use $\Omega \left( p^{-r}||x-a||_{p}\right) $ to denote the
characteristic function of the ball $B_{r}^{n}(a)$. We will use the notation
$1_{A}$ for the characteristic function of a set $A$. Along the article $%
d^{n}x$ will denote a Haar measure on $\mathbb{Q}_{p}^{n}$ normalized so that $\int_{\mathbb{Z}_{p}^{n}}d^{n}x=1.$

\subsection{Some function spaces}

A complex-valued function $\varphi $ defined on $\mathbb{Q}_{p}^{n}$ is \textit{called locally constant} if for any $x\in\mathbb{Q}_{p}^{n}$ there exist an integer $l(x)\in \mathbb{Z}$ such that
\begin{equation*}
\varphi (x+x^{\prime })=\varphi (x)\text{ for }x^{\prime }\in B_{l(x)}^{n}.
\end{equation*}%
A function $\varphi : \mathbb{Q}_{p}^{n}\rightarrow \mathbb{C}$ is called a \textit{Bruhat-Schwartz function
(or a test function)} if it is locally constant with compact support. The $\mathbb{C}$-vector space of Bruhat-Schwartz functions is denoted by $\mathcal{D}:=\mathcal{D}(\mathbb{Q}_{p}^{n})$. Let $\mathcal{D}^{\prime }:=\mathcal{D}^{\prime }(\mathbb{Q}_{p}^{n})$ denote the set of all continuous functional (distributions) on $%
\mathcal{D}$. The natural pairing $\mathcal{D}^{\prime }(\mathbb{Q}_{p}^{n})\times \mathcal{D}(\mathbb{Q}_{p}^{n})\rightarrow \mathbb{C}$ is denoted as $\left( T,\varphi \right) $
for $T\in \mathcal{D}^{\prime }(\mathbb{Q}_{p}^{n})$ and $\varphi \in \mathcal{D}(\mathbb{Q}_{p}^{n})$, see e.g. \cite[Section 4.4]{Albeverio et al}.

Every $f\in $ $L_{loc}^{1}(\mathbb{Q}_{p}^{n})$ defines a distribution $f\in \mathcal{D}^{\prime }\left(\mathbb{Q}_{p}^{n}\right) $ by the formula
\begin{equation*}
\left( f,\varphi \right) =\int\limits_{\mathbb{Q} _{p}^{n}}f\left( x\right) \varphi \left( x\right) d^{n}x.
\end{equation*}%
Such distributions are called \textit{regular distributions}.

We will denote by $\mathcal{D}_{\mathbb{R}}:=\mathcal{D}_{\mathbb{R}}( \mathbb{Q}_{p}^{n})$, the $\mathbb{R}$-vector space of test functions, and by $\mathcal{D}_{\mathbb{R}}^{\prime }:=\mathcal{D}_{\mathbb{R}}^{\prime }(\mathbb{Q}_{p}^{n})$, the $\mathbb{R}$-vector space of distributions.

Given $\rho \in \lbrack 0,\infty )$, we denote by $L^{\rho }:=L^{\rho }\left(\mathbb{Q}_{p}^{n}\right) :=L^{\rho }\left(\mathbb{Q}_{p}^{n},d^{n}x\right) ,$ the $\mathbb{C}-$vector space of all the complex valued functions $g$ satisfying $\int_{\mathbb{Q}_{p}^{n}}\left\vert g\left( x\right) \right\vert ^{\rho }d^{n}x<\infty $, $L^{\infty }\allowbreak :=L^{\infty }\left(\mathbb{Q}
_{p}^{n}\right) =L^{\infty }\left(\mathbb{Q}_{p}^{n},d^{n}x\right) $ denotes the $\mathbb{C}-$vector space of all the complex valued functions $g$ such that the
essential supremum of $|g|$ is bounded. The corresponding $\mathbb{R}$%
-vector spaces are denoted as $L_{\mathbb{R}}^{\rho }\allowbreak :=L_{\mathbb{R}}^{\rho }\left( 
\mathbb{Q}_{p}^{n}\right) =L_{\mathbb{R}}^{\rho }\left( \mathbb{Q}_{p}^{n},d^{n}x\right) $, 1$\leq \rho \leq \infty $.

Let denote by $C_{\mathbb{C}}:=C(\mathbb{Q}_{p}^{n},\mathbb{C})$ the $\mathbb{C}-$vector space of all complex valued functions which are continuous, by $C_{\mathbb{R}}:=C(\mathbb{Q}_{p}^{n},\mathbb{R})$ the $\mathbb{R}-$vector space of continuous functions.
Set
\begin{equation*}
C_{0}(\mathbb{Q}_{p}^{n},\mathbb{C}):=\left\{ f: \mathbb{Q}_{p}^{n}\rightarrow \mathbb{C};\text{ }f\text{ is continuous and }\lim_{||x||_{p}\rightarrow \infty
}f(x)=0\right\} ,
\end{equation*}%
where $\lim_{||x||_{p}\rightarrow \infty }f(x)=0$ means that for every $%
\epsilon >0$ there exists a compact subset $B(\epsilon )$ such that $%
|f(x)|<\epsilon $ for $x\in 
\mathbb{Q}
_{p}^{n}\backslash B(\epsilon ).$ We recall that $(C_{0}(
\mathbb{Q}
_{p}^{n},\mathbb{C}),||\cdot ||_{L^{\infty }})$ is a Banach space. The
corresponding $\mathbb{R}$-vector space will be denoted as $C_{0}(
\mathbb{Q}
_{p}^{n},\mathbb{R})$.

\subsection{Fourier transform}

Set $\chi _{p}(y)=\exp (2\pi i\{y\}_{p})$ for $y\in  \mathbb{Q}_{p}$. The map $\chi _{p}(\cdot )$ is an additive character on $\mathbb{Q}_{p}$, i.e. a continuous map from $\left( 
\mathbb{Q}_{p},+\right) $ into $S$ (the unit circle considered as multiplicative
group) satisfying $\chi _{p}(x_{0}+x_{1})=\chi _{p}(x_{0})\chi _{p}(x_{1})$,
$x_{0},x_{1}\in \mathbb{Q}_{p}$. The additive characters of $\mathbb{Q}_{p}$ form an Abelian group which is isomorphic to $\left(\mathbb{Q}_{p},+\right) $, the isomorphism is given by $\xi \rightarrow \chi _{p}(\xi x)$, see e.g. \cite[Section 2.3]{Albeverio et al}.

 Given $x=(x_{1},\dots ,x_{n}),$ $\xi =(\xi _{1},\dots ,\xi _{n})\in 
\mathbb{Q}_{p}^{n}$, we set $x\cdot \xi :=\sum_{j=1}^{n}x_{j}\xi _{j}$. If $f\in L^{1}$ its Fourier transform is defined by
\begin{equation*}
(\mathcal{F}f)(\xi )=\int_{\mathbb{Q}_{p}^{n}}\chi _{p}(\xi \cdot x)f(x)d^{n}x,\quad \text{for }\xi \in \mathbb{Q}_{p}^{n}.
\end{equation*}

We will also use the notation $\mathcal{F}_{x\rightarrow \xi }f$ and $\widehat{f}$\ for the Fourier transform of $f$. The Fourier transform is a linear isomorphism from $\mathcal{D}(\mathbb{Q}_{p}^{n})$ onto itself satisfying
\begin{equation}
(\mathcal{F}(\mathcal{F}f))(\xi )=f(-\xi ),  \label{FF(f)}
\end{equation}
for every $f\in \mathcal{D}(\mathbb{Q}_{p}^{n}),$ see e.g. \cite[Section 4.8]{Albeverio et al}. If $f\in L^{2},$
its Fourier transform is defined as
\begin{equation*}
(\mathcal{F}f)(\xi )=\lim_{k\rightarrow \infty }\int_{||x||\leq p^{k}}\chi
_{p}(\xi \cdot x)f(x)d^{n}x,\quad \text{for }\xi \in \mathbb{Q}_{p}^{n},
\end{equation*}
where the limit is taken in $L^{2}.$ We recall that the Fourier transform is
unitary on $L^{2},$ i.e. $||f||_{L^{2}}=||\mathcal{F}f||_{L^{2}}$ for $f\in
L^{2}$ and that (\ref{FF(f)}) is also valid in $L^{2}$, see e.g. \cite[Chapter $III$, Section 2]{Taibleson}.

The Fourier transform $\mathcal{F}\left[ T\right] $ of a distribution $T\in \mathcal{D}^{\prime }\left(\mathbb{Q}_{p}^{n}\right) $ is defined by
\begin{equation*}
\left( \mathcal{F}\left[ T\right] ,\varphi \right) =\left( T,\mathcal{F} \left[ \varphi \right] \right) \text{ for all }\varphi \in \mathcal{D}(\mathbb{Q}_{p}^{n})\text{.}
\end{equation*}%
The Fourier transform $T\rightarrow \mathcal{F}\left[ T\right] $ is a linear isomorphism from $\mathcal{D}^{\prime }\left(\mathbb{Q}_{p}^{n}\right) $\ onto itself. Furthermore, $T=\mathcal{F}\left[ \mathcal{F} 
\left[ T\right] \left( -\xi \right) \right] $. We also use the notation $\mathcal{F}_{x\rightarrow \xi }T$ and $\widehat{T}$ for the Fourier
transform of $T.$

\section{Preliminary Results-Elliptic Pseudo-differential Operators \label{preliminary_results}}

Along this article denote by $\mathbb{N} =\{1,2,\ldots \}$ the set of natural numbers.\ Next, we collect some results about of elliptic polynomials of degree $d.$ For more details the reader may
consult \cite{Zu}, \cite[Subsection 2.3]{Zu-lib1}. 

\begin{definition}
Let $f(\xi )\in \mathcal{\mathbf{\mathbb{Q}}}_{p}^{n}[\xi _{1},\ldots ,\xi _{n}]$ be a non-constant polynomial. We say that $f(\xi )$ is an elliptic polynomial of degree $d,$ if it satisfies: 
\begin{enumerate}[(i)]
\item $f(\xi )$ is a homogeneous polynomial of degree $d,$ and
\item $f(\xi )=0\Leftrightarrow \xi =0.$
\end{enumerate}
\end{definition}

\begin{lemma}
\cite[Lemma 1]{Zu}\label{lemma1} Let $f(\xi )\in \mathcal{\mathbf{%
\mathbb{Q} }}_{p}^{n}[\xi _{1},\ldots ,\xi _{n}]$ be an elliptic polynomial of degree $d.$ Then there exist positive constants $C_{0}=C_{0}(f),$ $C_{1}=C_{1}(f)$ such that
\begin{equation*}
C_{0}||\xi ||_{p}^{d}\leq |f(\xi )|_{p}\leq C_{1}||\xi ||_{p}^{d},\text{ for
every }\xi \in \mathcal{\mathbf{\mathbb{Q}}}_{p}^{n}.
\end{equation*}
\end{lemma}

\begin{remark}
\label{Nota_support} If $f(\xi )\in \mathcal{\mathbf{%
\mathbb{Q}
}}_{p}^{n}[\xi _{1},\ldots ,\xi _{n}]$ is an elliptic polynomial of degree $d,$ then for every $1\leq \rho <\infty $ and every $t>0,$ $e^{-t|f(\xi)|_{p}^{\beta }}\in L^{\rho }(\mathbb{Q}_{p}^{n})$. Indeed, by Lemma \ref{lemma1} we have
\begin{eqnarray*}
\int_{\mathbf{%
\mathbb{Q}
}_{p}^{n}}e^{-t\rho |f(\xi )|_{p}^{\beta }}d^{n}\xi  &\leq &\int_{\mathbf{%
\mathbb{Q}
}_{p}^{n}}e^{-t\rho C_{0}||\xi ||_{p}^{\beta d}}d^{n}\xi  \\
&\leq &\int_{\mathbf{ 
\mathbb{Z}
}_{p}^{n}}e^{-t\rho C_{0}||\xi ||_{p}^{\beta d}}d^{n}\xi +\int_{\mathbf{
\mathbb{Q}
}_{p}^{n}\backslash \mathbf{
\mathbb{Z}
}_{p}^{n}}e^{-t\rho C_{0}||\xi ||_{p}^{\beta d}}d^{n}\xi \  \\
&=&(1-p^{-n})\left( \sum_{j=0}^{\infty }\frac{p^{-nj}}{e^{t\rho
C_{0}p^{-j\beta d}}}+\sum_{i=1}^{\infty }\frac{p^{ni}}{e^{t\rho
C_{0}p^{i\beta d}}}\right) <\infty .
\end{eqnarray*}
\end{remark}

The elliptic pseudo-differential operators were introduced by Z\'{u}\~{n}iga-Galindo, see \cite{Zu-2003},\cite{Zu-2004}. Below we list some results about the heat \textit{Kernel} associated with these operators, and that will be of great importance throughout this article.

\begin{definition}
\label{def_operator}Let $f(\xi )\in \mathcal{\mathbf{\mathbb{Q}}}_{p}^{n}[\xi _{1},\ldots ,\xi_{n}]$ be a non-constant polynomial. An
operator of the form $(f(\mathcal{\partial },\beta )\varphi )(x):=\mathcal{F}_{\xi\rightarrow x }^{-1}(|f|_{p}^{\beta }\mathcal{F}_{x\rightarrow \xi }\varphi ),$ $\beta >0,$ $\varphi \in \mathcal{D}(\mathcal{\mathbf{\mathbb{Q}}}_{p}^{n}),$ is called a pseudo-differential operator with symbol $|f|_{p}^{\beta }:=|f(\xi )|_{p}^{\beta }.$

If $f(\xi )\in \mathcal{\mathbf{\mathbb{Z}}}_{p}^{n}[\xi _{1},\ldots ,\xi _{n}]$ is an elliptic polynomial of degree $d,$ then we say that $|f|_{p}^{\beta }$ is an elliptic symbol, and that $f(\mathcal{\partial },\beta )$ is an elliptic pseudo-differential operator of order $d.$
\end{definition}

\begin{remark}
$(i)$ \cite[Remark 1]{Zu} Note that $g(\xi _{1},\xi _{2})=\xi _{1}^{d}+p\xi_{2}^{d}$, $d\geq 2,$ is an elliptic polynomial of degree $d.$ Therefore, given $m\geq 1,$ there exists an elliptic polynomial in $m$ variables.

$(ii)$ \cite[Subsection 2.3]{Zu-lib1} There are infinitely many elliptic
polynomials and for any $n\in \mathbb{N}$ and $p\neq 2,$ there exists an elliptic polynomial $h(\xi _{1},\ldots ,\xi_{n})$ with coefficients in $\mathbb{Z}_{p}^{\times }$ and degree $2d(n):=2d$ such that
\begin{equation*}
|h(\xi _{1},\ldots ,\xi _{n})|_{p}=||(\xi _{1},\ldots ,\xi _{n})||_{p}^{2d}.
\end{equation*}
For the above note that the Taibleson operator is elliptic for $p\neq 2$.
Moreover, we can obtain pseudo-differential operators with radial symbols and with non-radial symbols.
\end{remark}

Throughout this paper (unless otherwise stated) we will assume that $f(\mathcal{\partial },\beta )$ is an elliptic pseudo-differential operator of order $d.$

\begin{remark}
Considering the Cauchy problem
\begin{equation}
\left\{
\begin{array}{ll}
\frac{\partial u}{\partial t}(x,t)=-(f(\mathcal{\partial },\beta )u)(x,t),%
\text{ \ } & t\in \lbrack 0,\infty ),\text{ \ }x\in \mathbb{Q}_{p}^{n} \\
&  \\
u(x,0)=u_{0}(x)\in \mathcal{D}(\mathcal{\mathbf{\mathbb{Q}}}_{p}^{n})\text{.} &
\end{array}%
\right.  \label{Cauchy_problem}
\end{equation}

We have that the function
\begin{equation*}
u(x,t):=\int\nolimits_{\mathbf{\mathbb{Q}}_{p}^{n}}\chi _{p}(-x,\xi )e^{-t|f(\xi )|_{p}^{\beta }}\widehat{u_{0}}(\xi)d^{n}\xi ;\text{ }u_{0}(x)\in \mathcal{D}(\mathbb{Q}_{p}^{n}),\ x\in
\mathbb{Q}_{p}^{n},\text{ }t\geq 0,
\end{equation*}%
satisfies the Cauchy problem (\ref{Cauchy_problem}).
\end{remark}

We define the \textit{heat Kernel} attached to operator $f(\mathcal{\partial},\beta )$ as
\begin{equation}
Z(x,t):=\mathcal{F}_{\xi\rightarrow x }^{-1}(e^{-t|f(\xi)|_{p}^{\beta }})=\int\nolimits_{\mathbf{
\mathbb{Q}}_{p}^{n}}\chi _{p}(-x,\xi )e^{-t|f(\xi )|_{p}^{\beta }}d^{n}\xi ,\text{ }
t>0,\text{ }x\in \mathbf{\mathbb{Q}}_{p}^{n}.  \label{def_Z(x,t)}
\end{equation}%
When considering $Z(x,t)$ as a function of $x$ for $t$ fixed, we will write $%
Z_{t}(x).$

\begin{lemma}
\label{properties_Z(x,t)}$(i)$ \cite[(P1)-Proposition 2]{Zu}$\int\nolimits_{\mathbf{\mathbb{Q}}_{p}^{n}}Z_{t}(x)d^{n}x=1$, for any $t>0$.

$(ii)$ \cite[(P3) Proposition 2]{Zu} $Z_{t+s}(x)=\int\nolimits_{\mathbf{\mathbb{Q}}_{p}^{n}}Z_{t}(x-y)Z_{s}(y)d^{n}y$, for $t,s>0,$ i.e. the heat Kernel satisfies the Chapman-Kolmogorov equation.

$(iii)$ \cite[Theorem 2]{Zu} For every $x\in \mathbf{
\mathbb{Q}
}_{p}^{n}$ and every $t>0$, we have that $Z(x,t)\geq 0.$
\end{lemma}

\begin{remark}
\label{Characte_Z}$(i)$ By Lemma \ref{properties_Z(x,t)}$-(iii),$ we have
that if $\varphi \in \mathcal{D}(
\mathbb{Q}
_{p}^{n})$ is a positive test function, then
\begin{equation*}
\int\nolimits_{\mathbf{
\mathbb{Q}
}_{p}^{n}}\chi _{p}(-x,\xi )e^{-t|f(\xi )|_{p}^{\beta }}\widehat{\varphi }%
(\xi )d^{n}\xi =(Z_{t}\ast \varphi )(x)\geq 0.
\end{equation*}%
$(ii)$ for any fixed $t>0$ and any $1\leq \rho <\infty ,$ $Z(x,t)\in L^{\rho
}(
\mathbb{Q}
_{p}^{n}).$ Moreover, $Z(x,t)$ is continuous function in $x,$ for any fixed $%
t>0,$ see \cite[Corollary 1]{Zu}. On the other hand, the solution the Cauchy
problem (\ref{Cauchy_problem}) satisfies%
\begin{equation}
u(x,t)=(Z_{t}\ast u_{0})(x),\text{ }u_{0}(x)\in \mathcal{D}(\mathbb{Q}_{p}^{n}),\ x\in 
\mathbb{Q}_{p}^{n},\text{ }t\geq 0,  \label{solution}
\end{equation}%
so that $u(\cdot ,t)$ is a continuous function for any $t>0.$

$(iii)$ As a direct consequence of \cite[Theorem 1 and 2]{Zu}, we have for
any $x\in \mathbf{\mathbb{Q}}_{p}^{n}$ and every $t>0,$ $Z(x,t)\leq At||x||_{p}^{-d\beta -n},$ where $A$ is a positive constant.
\end{remark}

\section{\label{Feller and Transition}Feller Semigroups and Markov
Transition Functions Associated to Heat Kernel}

The goal of this section is to prove that there are a Feller semigroup and a uniformly stochastically continuous $C_{0}-$transition function on $\mathbf{\mathbb{Q}}_{p}^{n}$ associated with the Heat Kernel attached to operator $f(\mathcal{\partial },\beta ).$ Significantly, we explicitly write the Feller semigroup and the Markov Transition Function.
\subsection{Feller Semigroups}

\begin{definition}
A one-parameter family $\{T_{t}\}_{t\geq 0}$ of bounded linear operators on $C_{0}(\mathbf{\mathbb{Q}}_{p}^{n})$ into itself is called a contraction semigroup if it satisfies
the following conditions:

\begin{enumerate}[(i)]
\item $T_{t+s}=T_{t}\cdot T_{s}$ for all $t,s\geq 0.$

\item $\lim_{t\rightarrow 0^{+}}||T_{t}u-u||_{L^{\infty }}=0$ for every $%
u\in C_{0}(\mathbf{\mathbb{Q}}_{p}^{n})$ (strongly continuous);

\item $||T_{t}||_{L^{\infty }}\leq 1$ for all $t\geq 0.$
\end{enumerate}
\end{definition}

For $u\in C_{0}(\mathbf{\mathbb{Q}}_{p}^{n}),$ $x\in \mathbf{\mathbb{Q}}_{p}^{n}$ and $t\geq 0$, we define%
\begin{equation}
T_{t}u(x):=\left\{
\begin{array}{ll}
u(x) & \text{if }t=0\text{, } \\
&  \\
\int_{\mathbf{\mathbb{Q}}_{p}^{n}}Z_{t}(x-y)u(y)d^{n}y=(Z_{t}\ast u)(x) & \text{if }t>0.%
\end{array}%
\right.  \label{Def_Tt}
\end{equation}

\begin{lemma}
\label{bounded_operator} For all $t\geq 0$,
\begin{equation*}
T_{t}:C_{0}(\mathbb{Q}_{p}^{n})\rightarrow C_{0}(\mathbb{Q}_{p}^{n})
\end{equation*}%
is a bounded linear operator with $||T_{t}||_{L^{\infty }}\leq 1$.
\end{lemma}

\begin{proof}
We consider the case $t>0$, since in the case $t=0,$ the assertion is clear.

Let $u\in C_{0}(\mathbf{\mathbb{Q}}_{p}^{n})$ and $x\in \mathbf{\mathbb{Q}}_{p}^{n}.$ Then, by Lemma \ref{properties_Z(x,t)}-$(i)$ we have \ \ \ \
\begin{equation}
|T_{t}u(x)|=\left\vert \int_{\mathbf{\mathbb{Q}}_{p}^{n}}Z_{t}(x-y)u(y)d^{n}y\right\vert \leq ||u||_{L^{\infty }}\int_{\mathbf{\mathbb{Q}}_{p}^{n}}Z_{t}(x-y)d^{n}y=||u||_{L^{\infty }}.  \label{cota_Tt}
\end{equation}

Moreover, since $Z_{t}(x)\in L^{1}(\mathbb{Q}_{p}^{n})$, $t>0,$ and $u$ is bounded, we have that $T_{t}u(x)=(Z_{t}\ast u)(x)$ is continuous.

On the other hand, if $||x||_{p}\gg 0$ and assuming without loss of
generality that $Supp(u)\subseteq B_{M}^{n},$ $M\in \mathbb{N} .$ Then\ by Remark \ref{Characte_Z}-$(iii)$ and the fact that $||\cdot ||_{p} $ is an ultranorm, we have
\begin{eqnarray*}
0 &\leq &\left\vert T_{t}u(x)\right\vert \leq ||u||_{L^{\infty
}}\int_{B_{M}^{n}}Z_{t}(x-y)d^{n}y\leq At||u||_{L^{\infty
}}\int_{B_{M}^{n}}||x-y||_{p}^{-d\beta -n}d^{n}y \\
&=&At||u||_{L^{\infty }}||x||_{p}^{-d\beta -n}Vol(B_{M}^{n})=0.
\end{eqnarray*}%
Therefore, the space $C_{0}(\mathbf{\mathbb{Q}}_{p}^{n})$ is an invariant subspace for the operators $T_{t}$, $t\geq 0,$ i.e.
\begin{equation*}
u\in C_{0}(\mathbf{\mathbb{Q}}_{p}^{n})\longrightarrow T_{t}u(x)\in C_{0}(\mathbb{Q}_{p}^{n}).
\end{equation*}
\end{proof}

\begin{lemma}
\label{semigroup} The family of operators $\{T_{t}\}_{t\geq 0}$\ defined in (\ref{Def_Tt}) determine a semigroup over the space $C_{0}(\mathbb{Q}_{p}^{n})$.
\end{lemma}

\begin{proof}
For $u\in C_{0}(\mathbf{\mathbb{Q}}_{p}^{n})$ and $t,s\geq 0$ we have that
\begin{eqnarray*}
T_{t}(T_{s}u)(x) &=&\int_{\mathbf{\mathbb{Q}}_{p}^{n}}Z_{t}(x-y)\left[ (T_{s}u)(y)\right] d^{n}y \\
&=&\int_{\mathbf{\mathbb{Q}}_{p}^{n}}\left[ \int_{\mathbf{\mathbb{Q}}_{p}^{n}}Z_{t}(x-y)Z_{s}(y-z)d^{n}y\right] u(z)d^{n}z \\
&=&\int_{\mathbf{
\mathbb{Q}
}_{p}^{n}}\left[ \int_{\mathbf{
\mathbb{Q}
}_{p}^{n}}Z_{t}((x-z)-w)Z_{s}(w)d^{n}w\right] u(z)d^{n}z \\
&=&\int_{\mathbf{
\mathbb{Q}
}_{p}^{n}}Z_{t+s}(x-z)u(z)d^{n}z=T_{t+s}u(x).
\end{eqnarray*}
\end{proof}

\begin{lemma}
\label{strongly_continuous}For all $u\in C_{0}(\mathbf{
\mathbb{Q}
}_{p}^{n})$ we have that
\begin{equation*}
\lim_{t\rightarrow 0^{+}}||T_{t}u-u||_{L^{\infty }}=0.
\end{equation*}
\end{lemma}

\begin{proof}
The desired equality follows from the following two statements:

\begin{claim}
\label{Claim 1} Let fixed $x\in \mathbf{
\mathbb{Q}
}_{p}^{n}.$ Then, for any given number $\epsilon >0$ we can find a $%
s:=s(x,\epsilon )\in
\mathbb{Z}
$ such that if $||x-y||_{p}<p^{s}$ then
\begin{equation*}
|I_{1}|=\left\vert
\int_{||x-y||<p^{s}}Z_{t}(x-y)[u(y)-u(x)]d^{n}y\right\vert <\epsilon .
\end{equation*}

The proof of the Claim is a direct consequence of Lemma \ref%
{properties_Z(x,t)}-$(i)$ taking into account that $u\in C_{0}(\mathbf{
\mathbb{Q}
}_{p}^{n})$.
\end{claim}

\begin{claim}
\label{Claim 2} Let fixed $x\in \mathbf{
\mathbb{Q}
}_{p}^{n}.$ Then, for any given number $\epsilon >0$ we can find a $%
s:=s(x,\epsilon )\in
\mathbb{Z}
$ such that if $||x-y||_{p}\geq p^{s}$ then
\begin{equation*}
I_{2}=\int_{||x-y||\geq p^{s}}Z_{t}(x-y)[u(y)-u(x)]d^{n}y\rightarrow 0,\text{
when }t\rightarrow 0^{+}.
\end{equation*}

The Claim's proof is as follows: Since that $u\in C_{0}(\mathbf{\mathbb{Q}}_{p}^{n})$, then for any number $\epsilon >0$, however small, there exists
some number $s:=s(x,\epsilon )\in \mathbb{Z}$ such that if $||x-y||_{p}<p^{s}$ then $||u(y)-u(x)||_{L^{\infty }}<\epsilon .$ Therefore, by Remark \ref{Characte_Z}-$(iii)$ we have \
\begin{eqnarray*}
\left\vert \int_{||x-y||\geq p^{s}}Z_{t}(x-y)[u(y)-u(x)]d^{n}y\right\vert
&\leq &2||u||_{L^{\infty }}\int_{||x-y||\geq p^{s}}Z_{t}(x-y)d^{n}y \\
&=&2||u||_{L^{\infty }}\int_{||w||\geq p^{s}}Z_{t}(w)d^{n}w \\
&\leq &2At||u||_{L^{\infty }}\int_{||w||\geq p^{s}}||w||_{p}^{-d\beta
-n}d^{n}w \\
&=&2Atp^{sd\beta }||u||_{L^{\infty }}\int_{||v||\geq 1}||v||_{p}^{-d\beta
-n}d^{n}v \\
&\leq &Ct||u||_{L^{\infty }}.
\end{eqnarray*}%

Therefore, By Claim \ref{Claim 1} and Claim \ref{Claim 2}, given any $%
\epsilon >0$ we have that
\begin{equation*}
\lim_{t\rightarrow 0^{+}}\sup |(T_{t}u)(x)-u(x)|\leq \lim_{t\rightarrow 0^{+}}\sup \left\vert |I_{1}|+|I_{2}|\right\vert \leq \epsilon ,\text{ for all }x\in \mathbf{\mathbb{Q}}_{p}^{n}.
\end{equation*}
\end{claim}
\end{proof}

\begin{definition}
\cite[Definition 12.8]{Berg-Gunnar} A strongly continuous contraction
semigroups $\{T_{t}\}_{t\geq 0}$ on $C_{0}(\mathbf{\mathbb{Q}}_{p}^{n}),$ for which all the operators $T_{t}$ are positive, i.e. such that for all $t>0$%
\begin{equation*}
u\in C_{0}(\mathbf{\mathbb{Q}}_{p}^{n})\text{ with }u\geq 0\text{ implies }T_{t}u\geq 0,
\end{equation*}%
is called a Feller semigroup on $\mathbf{\mathbb{Q}}_{p}^{n}.$\
\end{definition}

By Lemma \ref{properties_Z(x,t)}$-(iii)$ and (\ref{Def_Tt}), we have that if
$u\in C_{0}(\mathbf{\mathbb{Q}}_{p}^{n})$ and $u\geq 0,$ then $(T_{t}u)(x)\geq 0$, $t\geq 0.$ Moreover, as
a consequence of the previous lemma (Lemma \ref{bounded_operator}, Lemma \ref%
{semigroup}, Lemma \ref{strongly_continuous}), and (\ref{Def_Tt}), we have
that the family of operators $\{T_{t}\}_{t\geq 0}$ defined in (\ref{Def_Tt})
determine a strongly continuous contraction semigroup on $C_{0}(\mathbf{\mathbb{Q}}_{p}^{n}).$We have proved the following theorem

\begin{theorem}
\label{Feller_semigroups} The family of operators $\{T_{t}\}_{t\geq 0}$
defined in (\ref{Def_Tt}) determine a Feller semigroup on $\mathbf{\mathbb{Q}}_{p}^{n}.$
\end{theorem}

\begin{remark}
\label{Conservative} As a consequence of the previous theorem and Lemma \ref%
{properties_Z(x,t)}-$(i)$, the Feller semigroup $\{T_{t}\}_{t\geq 0}$ is
conservative, i.e. $T_{t}1=1,$ for all $t>0.$
\end{remark}

\subsection{Markov Transition Functions}

\begin{remark}
\label{Feller_Taira} Following the definition and terminology used on Feller
semigroups given in \cite[p. 43]{Taira}, we will show that our Feller
semigroup $\{T_{t}\}_{t\geq 0}$ satisfies the following properties: By Lemma \ref{bounded_operator} and Lemma \ref{strongly_continuous} we have that the
semigroup $\{T_{t}\}_{t\geq 0}$ is strongly continuous in $t$ for all $t\geq 0:$
\begin{equation*}
\lim_{s\rightarrow 0^{+}}||T_{t+s}f-T_{t}f||_{L^{\infty }}=0,\text{ }f\in C_{0}(\mathbf{\mathbb{Q}}_{p}^{n}).
\end{equation*}

Moreover, by Lemma \ref{properties_Z(x,t)}$-(i),$ Lemma \ref{properties_Z(x,t)}$-(iii)$ and (\ref{Def_Tt}) we have that the Feller semigroup $\{T_{t}\}_{t\geq 0}$ is non-negative and contractive on $C_{0}(\mathbf{\mathbb{Q}}_{p}^{n})$:%
\begin{equation*}
f\in C_{0}(\mathbf{\mathbb{Q}}_{p}^{n}),\text{ }0\leq f(x)\leq 1\text{ on }K\Longrightarrow \text{ }0\leq T_{t}f(x)\leq 1\text{ on }\mathbf{\mathbb{Q}}_{p}^{n}.
\end{equation*}
\end{remark}

\begin{definition}
\begin{enumerate}[(i)]
\item A function $p_{t}(x,E)$, defined for all $t\geq 0$, $x\in 
\mathbb{Q}_{p}^{n}$ and $E\in \mathcal{B}(\mathbb{Q}_{p}^{n})$, is called a Markov transition function on $\mathbb{Q}_{p}^{n}$ if it satisfies the following four conditions:
\begin{enumerate}[(a)]
\item $p_{t}(x,\cdot )$ is a measure on $\mathcal{B}(\mathbb{Q}_{p}^{n})$ and $p_{t}(x,\mathbb{Q}_{p}^{n})\leq 1$ for all $t\geq 0$ and $x\in \mathbb{Q}_{p}^{n}.$
\item $p_{t}(\cdot ,E)$ is a Borel measurable function for all $t\geq 0$ and $E\in \mathcal{B}(
\mathbb{Q}_{p}^{n})$.
\item $p_{0}(x,\{x\})=1$ for all $x\in \mathbf{\mathbb{Q}}_{p}^{n}.$
\item \textbf{(The Chapman-Kolmogorov equation)} For all $t,s\geq 0,$ $x\in
\mathbf{\mathbb{Q}}_{p}^{n}$ and $E\in \mathcal{B}(\mathbb{Q}_{p}^{n})$, we have the equations
\begin{equation*}
p_{t+s}(x,E)=\int_{\mathbf{\mathbb{Q}}_{p}^{n}}p_{t}(x,d^{n}y)p_{s}(y,E).
\end{equation*}
\end{enumerate}
\item We say that the Markov transition function $p_{t}(x,\cdot )$ on $
\mathbb{Q}_{p}^{n}$ satisfies the $condition$ $(L)$ if for each $s>0$ and each compact
subset $E\subset \mathbb{Q}_{p}^{n}$,%
\begin{equation*}
\lim_{x\rightarrow \infty }\sup_{0\leq t\leq s}p_{t}(x,E)=0.
\end{equation*}
\item A Markov transition function $p_{t}(x,\cdot )$ on $\mathbb{Q}_{p}^{n}$ is said to be uniformly stochastically continuous on $\mathbb{Q}_{p}^{n}$ if the following condition is satisfied:

For each $r\in \mathbb{Z}$ and each compact $E\subset \mathbb{Q}_{p}^{n}$, we have that
\begin{equation*}
\lim_{t\rightarrow 0^{+}}\sup_{x\in E}[1-p_{t}(x,B_{r}^{n}(x))]=0.
\end{equation*}
\item We say that $p_{t}(x,\cdot )$ is a $C_{0}-$function if the space $C_{0}(\mathbb{Q}_{p}^{n})$ is an invariant subspace for the operators $T_{t}:$ 
\begin{equation*}
f\in C_{0}(\mathbb{Q}_{p}^{n})\Longrightarrow T_{t}f\in C_{0}(\mathbb{Q}_{p}^{n}).
\end{equation*}
\end{enumerate}
\end{definition}

\begin{definition}
For $E\in \mathcal{B}(\mathbb{Q}_{p}^{n})$, we define%
\begin{equation}
p_{t}(x,E)=\left\{
\begin{array}{ll}
Z_{t}(x)\ast 1_{E}(x)\text{,} & \text{\ for }t>0\text{, }x\in \mathbf{
\mathbb{Q}
}_{p}^{n} \\
&  \\
1_{E}(x), & \text{for }t=0\text{, }x\in \mathbf{
\mathbb{Q}
}_{p}^{n}.%
\end{array}%
\right.  \label{def_p_t}
\end{equation}
\end{definition}

\begin{theorem}
\label{Transition} $p_{t}(x,\cdot )$ is a uniformly stochastically continuous $C_{0}-$transition function on $\mathbf{\mathbb{Q}}_{p}^{n}$, satisfy condition $(L)$ and the formula
\begin{equation*}
T_{t}f(x):=\int_{\mathbf{\mathbb{Q}}_{p}^{n}}p_{t}(x,d^{n}y)f(y)
\end{equation*}%
holds. Moreover, it is the transition function of some strong Markov
processes $\mathfrak{X}$ with state space $\mathbb{Q}_{p}^{n}$ and transition function $p_{t}(x,\cdot )$ whose paths are right continuous and have no discontinuities other than jumps.
\end{theorem}

\begin{proof}
The results are a consequence of Remark \ref{Feller_Taira} and Theorem \ref{Feller_semigroups} by well-known results in the theory of Markov processes,
see e.g. \cite[Theorem 2.15]{Taira} and \cite[Theorem 2.12]{Taira}. 
\end{proof}

\begin{remark}
\begin{enumerate}[(i)]
\item  Significantly, in the previous theorem we explicitly write the
transition function $p_{t}(x,\cdot )$ on $\mathcal{B}(\mathbb{Q}_{p}^{n}).$ Unlike of \cite[Remark 40]{Zu-lib1}, see also \cite[proof of Lemma 5]{To-Z}, where the existence of a Markov transition function on $\mathcal{B}(\mathbb{Q}_{p}^{n})$ is implicitly guaranteed.

\item   As for any $x\in \mathbf{\mathbb{Q}}_{p}^{n}$ and every $t>0,$ $Z(x,t)\leq At||x||_{p}^{-d\beta -n},$ where $A$ is a positive constant, see Remark \ref{Characte_Z}$-(iii)$, we can show that, for each $r\in \mathbb{Z}$ and each compact $E\subset \mathbb{Q}_{p}^{n}$, we have the condition
\begin{equation*}
\lim_{t\rightarrow 0^{+}}\sup_{x\in E}p_{t}(x,\mathbb{Q}_{p}^{n}\backslash B_{r}^{n}(x))=0,
\end{equation*}

so that by the previous theorem and \cite[Theorem 2.10]{Taira} the paths of the strong Markov process $\mathfrak{X}$ are right continuous on $[0,\infty) $ and have left-hand limits on $[0,\infty )$ almost surely. On the other hand, by using \cite[Sect. 2]{Evans}, it is possible to show that the process $\mathfrak{X}$ constructed in the previous theorem is an L\'{e}vy processes with state space $\mathbb{Q}_{p}^{n}$ and transition function $p_{t}(x,\cdot ).$
\end{enumerate}
\end{remark}

\section{\label{Negative_definite_function}Symbols of Elliptic
Pseudo-differential Operators and Negative Definite Functions}

The goal of this section is show that the symbol of elliptic pseudo-differential operators $(|f(\xi )|_{p}^{\beta })$ is a function negative definite function. Moreover, we show that this symbol can be represented as a combination of a constant $c\geq 0,$ a continuous
homomorphism $l: \mathbb{Q}_{p}^{n}\rightarrow \mathbb{R}$ and a non-negative, continuous quadratic form $q: \mathbb{Q} _{p}^{n}\rightarrow \mathbb{R}.$

\begin{definition}
\label{negative_definite} A function $f:\mathbb{\mathbb{Q}}_{p}^{n}\rightarrow \mathbb{C}$ is called negative definite if  
\begin{equation*}
\sum\nolimits_{i,j=1}^{m}\left( f(x_{i})+\overline{f(x_{j})} -f(x_{i}-x_{j})\right) \lambda _{i}\overline{\lambda _{j}}\geq 0
\end{equation*}%
for all $m\in \mathbb{N}\backslash \{0\},$ $x_{1},\ldots ,x_{m}\in $ $\mathbb{\mathbb{Q}}_{p}^{n},$ $\lambda _{1},\ldots ,\lambda _{m} \in \mathbb{C}$.
\end{definition}

\begin{definition}
\label{def_convolution_semigroup}\cite[Definition 8.1]{Berg-Gunnar} A family
$(\mu _{t})_{t>0}$ of positive bounded measures on $\mathbb{Q}_{p}^{n}$ with the properties
\begin{enumerate}[(i)]
\item $\mu _{t}(\mathbb{Q}_{p}^{n})\leq 1$ for $t>0,$
\item $\mu _{t}\ast \mu _{s}=\mu _{t+s}$ for $t,s>0,$
\item $\lim_{t\rightarrow 0^{+}}\mu _{t}=\delta _{0}$ vaguely ($\delta
_{0} $ denotes the Dirac measure at $0\in \mathbb{Q}_{p}^{n})$, is called a convolution semigroup on $\mathbb{Q}_{p}^{n}.$
\end{enumerate}
\end{definition}

\begin{remark}
\label{Zt_convolution_semigroup} A Feller semigroup $\{T_{t}\}_{t>0}$ on $%
\mathbf{\mathbb{Q}}_{p}^{n}$ is said to be translation invariant if all the operators $T_{t}$
commute with the translations of $\mathbf{\mathbb{Q}}_{p}^{n},$ i.e. if
\begin{equation*}
T_{t}(\tau _{a}f)=\tau _{a}(T_{t}f)\text{ for }a\in \mathbf{\mathbb{Q}}_{p}^{n},\text{ }t>0\text{ and }f\in C_{0}(\mathbf{\mathbb{Q}}_{p}^{n}),
\end{equation*}%
where $\tau _{a}f$ denotes the function
\begin{equation*}
\tau _{a}f(x)=f(x-a)\text{ for }x\in \mathbf{\mathbb{Q}}_{p}^{n}.
\end{equation*}%
It is easy to check that the Feller semigroup $\{T_{t}\}_{t>0}$ is translation invariant.
\end{remark}

\begin{theorem}
\label{negative_function} The symbol $|f|_{p}^{\beta },$ $\beta >0,$ of the pseudo-differential operator given in Definition \ref{def_operator} is a negative definite function.
\end{theorem}

\begin{proof}
By Remark \ref{Zt_convolution_semigroup} and \cite[Exercise 12.10]{Berg-Gunnar}, there exists a uniquely determined convolution semigroup $(\mu _{t})_{t>0}$ on $\mathbf{\mathbb{Q}}_{p}^{n}$ such that $\{T_{t}\}_{t>0}$ is given in terms of $(\mu_{t})_{t>0} $ by the formula
\begin{equation*}
T_{t}f=\mu _{t}\ast f\text{ for }f\in C_{0}(\mathbf{\mathbb{Q}}_{p}^{n})\text{ and }t>0.
\end{equation*}%
Therefore, by (\ref{Def_Tt}) we have that $Z_{t}=\mu _{t},$ $t>0,$ i.e. $(Z_{t})_{t>0}$ is a convolution semigroup on $\mathbf{\mathbb{Q}}_{p}^{n}.$

On the other hand, since there is a one-to-one correspondence between
convolution semigroups on $\mathbf{\mathbb{Q}}_{p}^{n}$ and continuous, negative definite functions on $\mathbf{\mathbb{Q}}_{p}^{n}$, see e.g. \cite[Theorem 8.3]{Berg-Gunnar}, we have by Remark \ref{Nota_support}, Lemma \ref{properties_Z(x,t)}$-(i)$ and (\ref{def_Z(x,t)})
that $|f|_{p}^{\beta },$ $\beta >0,$ is a negative definite function, where
\begin{equation*}
Z_{t}(x)=\int\nolimits_{\mathbf{\mathbb{Q}}_{p}^{n}}\chi _{p}(-x,\xi )e^{-t|f(\xi )|_{p}^{\beta }}d^{n}\xi ,\text{ }%
t>0,\text{ }x\in \mathbf{\mathbb{Q}}_{p}^{n}.
\end{equation*}
\end{proof}

\begin{remark}
\label{Obs_Z} In the proof of previous theorem it was obtained that $(Z_{t})_{t>0}$ is a convolution semigroup on $\mathbf{\mathbb{Q}}_{p}^{n}.$ Therefore, by (\ref{Def_Tt}), (\ref{def_p_t}) and Theorem \ref{Transition} we have that $Z_{t}(x-y)d^{n}y=p_{t}(x,\cdot ).$ Moreover, we
will say that $(Z_{t})_{t>0}$ and $|f(\xi )|_{p}^{\beta }$ are associated, see e.g. \cite[Definition 8.5]{Berg-Gunnar}.
\end{remark}

The infinitesimal generator $(A,D)$ of the semigroup $(P_{t})_{t>0}$ on $E$ is defined by
\begin{equation*}
Af=\lim_{t\rightarrow 0}t^{-1}(P_{t}f-f)\text{ for }f\in D,
\end{equation*}%
where $D$ is the set of elements in $E$ such that this limit exists in $E.$

Let $(A_{0},D_{0})$ (respectively $(A_{b},D_{b}),$ respectively $(A_{2},D_{2})$) denote the infinitesimal generator for $(Z_{t})_{t>0}.$

\begin{lemma}
\label{Lemma_Forst}\cite[Lemma 2]{Forst} The domain $D_{2}$ of $A_{2}$ is
given by
\begin{equation*}
D_{2}=\left\{ f\in L^{2}(\mathbf{\mathbb{Q}}_{p}^{n}):|f(\xi )|_{p}^{\beta }\widehat{f}\in L^{2}(\mathbf{\mathbb{Q}}_{p}^{n})\right\}
\end{equation*}%
and $(A_{2}f)^{\wedge }=-|f(\xi )|_{p}^{\beta }\widehat{f}$ for $f\in D_{2}.$
\end{lemma}

\begin{definition}
\cite[Definition 3]{Forst} The semigroup $(Z_{t})_{t>0}$ determines a sesquilinear form $\beta :D_{2}\times D_{2}\rightarrow \mathbb{C}$ by the definition
\begin{equation*}
\beta (f,g)=(-A_{2}f,g)\text{ for }f,g\in D_{2},
\end{equation*}%
and by Lemma \ref{Lemma_Forst} we can write
\begin{equation*}
\beta (f,g)=\left( |f(\xi )|_{p}^{\beta }\widehat{f},\widehat{g}\right)
\text{ for }f,g\in D_{2}.
\end{equation*}
\end{definition}

\begin{definition}
\cite[Definition 10]{Forst} The semigroup $(Z_{t})_{t>0}$ is said to be of
local type, if the conditions
\begin{enumerate}[(i)]
\item For all $f\in D_{b}:supp$ $A_{b}f\subseteq supp$ $f;$
\item For all $f\in D_{0}:supp$ $A_{0}f\subseteq supp$ $f;$
\item For all $f\in D_{2}:supp$ $A_{2}f\subseteq supp$ $f,$
\end{enumerate}
are fulfilled.
\end{definition}

\begin{remark}
\label{probability_measures} By Lemma \ref{properties_Z(x,t)}$-(i)$ and Remark \ref{Obs_Z}\ we have the family $(Z_{t})_{t>0}$ is a convolution semigroup consistent of probability measures on $\mathbb{Q}_{p}^{n}.$
\end{remark}

\begin{lemma}
\label{sesquilinear} $\beta (f,g)=0$ for all $g,h\in \mathcal{K}\cap D_{2}$
such that $g$ is constant in a neighbourhood of the support of $g.$ Here $\mathcal{K=K}(\mathbb{Q}_{p}^{n})$ be the space of complex, continuous functions on $
\mathbb{Q}_{p}^{n}$ with compact support, equipped with the usual topology.
\end{lemma}

\begin{proof}
Suppose that $D\subset \mathbb{Q}_{p}^{n}$ is a neighbourhood of the support of $h$ where $g$ is equal a constant $k,$ i.e. there is an open subset $U\subset D$ such that $supp(h)\subset U$ and $g|_{D}=k$. Using Fubini's Theorem and the Parseval-Steklov equality, see \cite[Theorem 5.3.1]{Albeverio et al}, we have
\begin{eqnarray*}
\beta (g,h) &=&\left( |f(\xi )|_{p}^{\beta }\widehat{g},\widehat{h}\right)
=\int\nolimits_{
\mathbb{Q}
_{p}^{n}}|f(\xi )|_{p}^{\beta }\widehat{g}(\xi )\widehat{\overline{h}}(\xi
)d^{n}\xi \\
&=&\int\nolimits_{
\mathbb{Q}
_{p}^{n}}|f(\xi )|_{p}^{\beta }\left( \int\nolimits_{
\mathbb{Q}
_{p}^{n}}\chi _{p}(\xi ,x)(g\ast \overline{h})(x)d^{n}x\right) d^{n}\xi \\
&=&\left( k\int\nolimits_{D}\overline{h}(y)d^{n}y\right) \left[
\int\nolimits_{
\mathbb{Q}
_{p}^{n}}|f(\xi )|_{p}^{\beta }\left( \int\nolimits_{
\mathbb{Q}
_{p}^{n}}\chi _{p}(\xi ,x)d^{n}x\right) d^{n}\xi \right] .
\end{eqnarray*}%
For $\xi \in \mathbb{Q} _{p}^{n}$ we have: If $\xi =0$ then $|f(\xi )|_{p}^{\beta }=0$, and if $\xi
\in \mathbb{Q}_{p}^{n}\backslash \{0\},$ then by \cite[Example 9-p. 44]{V-V-Z} we have
that $\int\nolimits_{\mathbb{Q}_{p}^{n}}\chi _{p}(\xi ,x)d^{n}x=0.$

From the above we have that $\beta (g,h)=0.$
\end{proof}

\begin{lemma}
\label{local_type} The convolution semigroup $(Z_{t})_{t>0}$ is of local type.
\end{lemma}

\begin{proof}
The result is followed by Lemma \ref{sesquilinear} and Remark \ref{probability_measures}, taking into account \cite[Remark 13-(a)]{Forst}.
\end{proof}

By \cite[Proposition 9]{Forst},\ \cite[Definition 18.24]{Berg-Gunnar} and
\cite[Exercise 18.26]{Berg-Gunnar}, the whole theory of convolution
semigroups of the local type can be adapted to the theory of the text \cite[Chapter III-\S\ 18]{Berg-Gunnar}.

\begin{lemma}
\label{Existencia_Levy}\cite[Proposition 18.2]{Berg-Gunnar} The net $\left(\frac{1}{t}Z_{t}|_{\mathbb{Q}_{p}^{n}\backslash \{0\}}\right) _{t>0}$ converges vaguely as $t\rightarrow 0^{+}$ to a measure $Z$ on $\mathbb{Q}_{p}^{n}\backslash \{0\}.$
\end{lemma}

\begin{definition}
\label{Def_Levy_measure}\cite[Definition 18.3]{Berg-Gunnar} The positive measure $Z$ on $\mathbb{Q}_{p}^{n}\backslash \{0\}$ from Lemma \ref{Existencia_Levy} is called the L\'{e}vy measure for the convolution semigroup $(Z_{t})_{t>0}$ on $\mathbb{Q}_{p}^{n}$ (and also the L\'{e}vy measure for the continuous negative function $|f|_{p}^{\beta }$ on $\mathbb{Q}_{p}^{n}$).
\end{definition}

\begin{definition}
\cite[Definition 7.18]{Berg-Gunnar} A function $q: \mathbb{Q}_{p}^{n}\rightarrow \mathbb{R}$ is called a quadratic form, if it satisfies the equation
\begin{equation*}
2q(x)+2q(\xi )=q(x+\xi )+q(x-\xi )\text{ \ for all }x,\xi \in \mathbb{Q}_{p}^{n}.
\end{equation*}
\end{definition}

\begin{remark}
\label{Obs_quadratic} It is easy to see that a quadratic form $q$ satisfies
\begin{eqnarray*}
q(0) &=&0, \\
q(x) &=&q(-x)\text{ for all }x\in 
\mathbb{Q}
_{p}^{n}, \\
q(nx) &=&n^{2}q(x)\text{ \ for all }x\in 
\mathbb{Q}
_{p}^{n}\text{ and }n\in 
\mathbb{N}.
\end{eqnarray*}%
Moreover, by \cite[Proposition 7.19]{Berg-Gunnar} we have that a
non-negative quadratic form $q$ on $\mathbb{Q}_{p}^{n}$ is negative definite.
\end{remark}

As a direct consequence of \cite[Theorem 18.27]{Berg-Gunnar},\ Lemma \ref{local_type}, Lemma \ref{Existencia_Levy} and Definition \ref{Def_Levy_measure} we have the following theorem.

\begin{theorem}
\label{equivalencias}The following conditions are equivalent:
\begin{enumerate}[(i)]
\item for all open neighbourhoods $W$ of $0$ we have
\begin{equation*}
\lim_{t\rightarrow 0^{+}}\frac{1}{t}Z_{t}(\complement W)=0.
\end{equation*}
\item $Z=0.$
\item $|f(\xi )|_{p}^{\beta }=c+il(\xi )+q(\xi )$ for $\xi \in 
\mathbb{Q}_{p}^{n}$, where $c\geq 0,$ $l: \mathbb{Q}_{p}^{n}\rightarrow \mathbb{R}$ is a continuous homomorphism and $q: \mathbb{Q}_{p}^{n}\rightarrow \mathbb{R}$ is a non-negative, continuous quadratic form.
\end{enumerate}
\end{theorem}


\begin{thebibliography}{99}
\bibitem{Albeverio et al} Albeverio S., Khrennikov A. Yu., Shelkovich V. M.,
Theory of $p$-adic distributions: linear and nonlinear models. London
Mathematical Society Lecture Note Series, 370. Cambridge University Press,
Cambridge, 2010.

\bibitem{Albeverio 2006} Albeverio S., Khrennikov A. Yu., Shelkovich V. M.,
Harmonic analysis in the $p-$adic Lizorkin spaces: Fractional operators,
pseudo-differential equations, $p-$adic wavelets, Tauberian theorems. J.
Fourier Anal. Appl. 12:4 (2006), 393-425. MR 2007g:47076.

\bibitem{Berg-Gunnar} Berg Christian, Forst Gunnar, Potential theory on
locally compact abelian groups. Springer-Verlag, New York-Heidelberg, 1975.

\bibitem{Chuong-Co-2008} Chuong N. M., Co N. V., The Cauchy problem for a
class of pseudo-differential equations over $p-$adic field. J. Math. Anal.
Appl. 340: 1 (2008), 629-645. MR 2009h: 35475 Zbl 1153.35095.

\bibitem{Evans} Evans, S. N., Local properties of L\'{e}vy processes on a
totally disconnected group. J. Theor. Probab. 2(2), 209--259 (1989)

\bibitem{Forst} Forst G., Convolution semigroups of local type. Math. Scand.
34, 211-218 (1974).

\bibitem{Hoh-1998} Hoh W., A symbolic calculus for pseudo differential
operators generating Feller semigroups. Osaka J. Math. 35 (1998), 789-820.

\bibitem{Hoh-Libro} Hoh W., Pseudo differential operators generating Markov
processes, Habilitationsschrift, Universit\"{a}t Bielefeld (1998).

\bibitem{Jacob-1994} Jacob N., A class of Feller semigroups generated by
pseudo differential operators. Math. Z. 215, 151-166 (1994).

\bibitem{Jacob-1992} Jacob N., Feller semigroups, Dirichlet forms and pseudo
differential operators. Forum Math. 4, 433-446 (1992).

\bibitem{Jacob-1993} Further pseudodifferential operators generating Feller
semigroups and Dirichlet forms. Revista Matem\'{a}tica Iberoamericana, Vol.
9, Nro. 2, (1993).

\bibitem{Jacob-vol-1} Jacob N., Pseudo differential operators and Markov
processes. Vol. I. Fourier analysis and semigroups. Imperial College Press,
London, 2001.

\bibitem{Jacob-vol-2} Jacob N., Pseudo differential operators and Markov
processes. Vol. II. Generators and their potential theory. Imperial College
Press, London, 2002.

\bibitem{Jacob-vol-3} Jacob N., Pseudo differential operators and Markov
processes. Vol. III. Markov processes and applications. Imperial College
Press, London, 2005.

\bibitem{Khrennikov-1992} Khrennikov A. Y., Fundamental solutions over the
field of $p-$adic numbers. Algebra i Analiz 4:3 (1992), 248-266. In Russian;
translated in St. Petersburg Math. J. 4:3 (1993), 613-628. MR 93k:46060 Zbl
0828.35003.

\bibitem{Kochubei-2008} Kochubei A. N., A non-Archimedean wave equation.
Pacific J. Math. $235:$2 (2008), 245-261. MR 2009e:35305 Zbl 05366328.

\bibitem{Kochubei-1993} Kochubei A. N., A Schr\"{o}dinger-type equation over
the field of $p-$adic numbers. J. Math. Phys. $34$:8 (1993), 3420-3428. MR
94f:47059 Zbl 0780.35088.

\bibitem{Kochubei-1998} Kochubei A. N., Fundamental solutions of
pseudodifferential equations associated with $p-$adic quadratic forms. Izv.
Ross. Akad. Nauk Ser. Mat. $62$:6 (1998), 103-124. In Russian; translated in
Izvestiya Math., $62$:6 (1998), 1169-1188. MR 2000f:11158 Zbl 0934.35218.

\bibitem{Kochubei-1991} Kochubei A. N., Parabolic equations over the field
of $p-$adic numbers. Izv. Akad. Nauk SSSR Ser. Mat. $55$:6 (1991),
1312-1330. In Russian; translated in Math. USSR Izvestiya $39$ (1992),
1263-1280. MR 93e:35050.

\bibitem{Kochubei-2001} Kochubei A. N., Pseudo-differential equations and
stochastic over non-Archimedean fields, Pure and Applied Mathematics $244,$
Marcel Dekker, New York, 2001. MR 2003b:35220 Zbl 0984.11063.

\bibitem{Ch-Z-1} Chac\'{o}n-Cortes L. F., Z\'{u}\~{n}iga-Galindo W. A.,
Nonlocal operators, parabolic-type equations, and ultrametric random walks.
J. Math. Phys. 54, 113503 (2013) $\And $ Erratum 55 (2014), no. 10, 109901,
1 pp.

\bibitem{R-Zu-2010} Rodr\'{\i}guez-Vega J. J., Z\'{u}\~{n}iga-Galindo W. A.,
Elliptic pseudodifferential equations and Sobolev spaces over $p-$adic
fields. Pacif. J. Math. $246$ (2010), 407-420.

\bibitem{R-Zu} Rodr\'{\i}guez-Vega J. J., Z\'{u}\~{n}iga-Galindo W. A.,
Taibleson operators, $p$-adic parabolic equations and ultrametric diffusion,
Pacific J. Math. 237 (2), 327--347 (2008).

\bibitem{Taibleson} Taibleson M. H., Fourier analysis on local fields.
Princeton University Press, 1975.

\bibitem{Taira} Taira Kazuaki, Boundary value problems and Markov processes.
Second edition. Lecture Notes in Mathematics, 1499. Springer-Verlag, 2009.

\bibitem{To-Z-2} Torresblanca-Badillo A., Z\'{u}\~{n}iga-Galindo W. A.,
Non-Archimedean Pseudodifferential Operators and Feller Semigroups, p-Adic
Numbers, Ultrametric Analysis and Applications, Vol. 10, No. 1, pp. 60-76,
2018.

\bibitem{To-Z} Torresblanca-Badillo A., Z\'{u}\~{n}iga-Galindo W. A.,
Ultrametric Diffusion, exponential landscapes, and the first passage time
problem, Acta Appl Math (2018), 157:93.

\bibitem{V-V-Z} Vladimirov V. S., Volovich I. V., Zelenov E. I., $p$-adic
analysis and mathematical physics. World Scientific, 1994.

\bibitem{Zu-2003} Z\'{u}\~{n}iga-Galindo W. A., Fundamental solutions of
pseudo-differential operators over $p-$adic fields. Rend. Sem. Mat. Univ.
Padova $109$, 241-245 (2003).

\bibitem{Zu} Z\'{u}\~{n}iga-Galindo W. A., Parabolic equations and Markov
processes over $p$-adic fields, Potential Anal. $28:2$, 185--200, (2008).

\bibitem{Zu-2004} Z\'{u}\~{n}iga-Galindo W. A., Pseudo-differential
equations connected with $p-$adic forms and local zeta functions. Bull.
Austral. Math. Soc. $70$:1, 73-86, (2004).

\bibitem{Zu-lib1} Z\'{u}\~{n}iga-Galindo W. A., Pseudodifferential Equations
Over Non-Archimedean Spaces. Lecture Notes in Mathematics 2174, Springer
International Publishing, 2016.
\end{thebibliography}
\end{document}